\documentclass[11pt,a4paper]{article}
\usepackage[utf8x]{inputenc}
\usepackage[T1]{fontenc}

\usepackage{amsmath}
\usepackage{amsthm}
\usepackage{comment}
\usepackage{amssymb}
\usepackage{cite}

\newcommand{\non}{\nonumber\\}

\usepackage[pdftex]{graphicx} 
\usepackage[pdftex,linkcolor=black,pdfborder={0 0 0}]{hyperref} 
\usepackage{calc} 
\usepackage{enumitem} 

\usepackage[a4paper, lmargin=0.1666\paperwidth, rmargin=0.1666\paperwidth, tmargin=0.1111\paperheight, bmargin=0.1111\paperheight]{geometry} 

\usepackage[all]{nowidow} 

\usepackage{lipsum} 

\usepackage{tikz}
\usetikzlibrary{decorations.markings}
\tikzset{->-/.style={decoration={
  markings,
  mark=at position #1 with {\arrow{>}}},postaction={decorate}}}
\tikzset{-<-/.style={decoration={
  markings,
  mark=at position #1 with {\arrow{<}}},postaction={decorate}}}

\newtheorem{theorem}{Theorem}
\newtheorem{lemma}[theorem]{Lemma}

\def\Real{{\mathbb R}}

\def\norm#1{\left|#1\right|}
\def\innerprod(#1,#2){{\left<#1\,,\,#2\right>}}

\usepackage{fancyhdr}
\fancypagestyle{firststyle}
{
   \fancyhf{}
   
\fancyfoot[C]{1}
\fancyfoot[R]{\footnotesize\bf File: \jobname.tex}
}
\hypersetup{ 	
pdfsubject = {},
pdftitle = {},
pdfauthor = {}
}

\newcommand{\ga}{\gamma}

\newcommand{\p}{\partial}

\def\iMa{{{\mu}}}
\def\iMb{{{\nu}}}
\def\iMc{{{\rho}}}
\def\iMd{{{\kappa}}}

\def\iSa{{{a}}}
\def\iSb{{{b}}}
\def\iSc{{{c}}}

\def\partz{\partial^{(z)}}

\def\partz{\partial}

\def\DEfullstop{\,.}
\def\DEcomma{\,,}
\def\DEnone{}

\def\Interval{{\cal I}}
\def\Mman{{\cal M}}

\def\Cdot{{\dot C}}

\def\ToneTen{{{\cal \theta}}}

\def\deltaThree{\delta^{(3)}}
\def\Vz{{\underline z}}
\def\Vw{{\underline w}}
\def\Vzero{{\underline 0}}
\def\calT{{\cal T}}
\def\TReg{{\cal T}}
\def\Sym{{\textup{Sym}}}
\def\calF{{\cal F}}
\def\dual#1{{\widetilde{#1}}}

\def\Bnabla{{\boldsymbol\nabla}}

\def\PF{{\varsigma}}
\allowdisplaybreaks

\begin{document} 

\title{The stress-energy distributional multipole for both uncharged and charged dust}
\author{
Jonathan Gratus$^{1,4,5}$,
Spyridon Talaganis$^{2,4}$,
and Willow Sparks$^{3,4,6}$
}
\maketitle
\noindent 
$^1$ j.gratus@lancaster.ac.uk. ORCID: 0000-0003-1597-6084
\\
$^2$ s.talaganis@lancaster.ac.uk. ORCID: 0000-0003-0113-7546
\\
$^3$ willow.sparks@stfc.ac.uk. ORCID:0009-0008-7630-9609
\\
$^{4}$ Dept Physics, Lancaster University, Lancaster LA1 4YB, UK
\\
$^{5}$ Cockcroft Institute of Accelerator Science, Daresbury Laboratory, Keckwick Lane, Daresbury, Warrington,
WA4 4AD, UK
\\
$^{6}$ Now based at the Scientific Computing Department, STFC, Daresbury Laboratory, Keckwick Lane, Daresbury, WA4 4AD, UK

\begin{abstract}
In this paper, we formulate the distributional uncharged and charged stress-energy tensors. These are integrals, along a worldline, of derivatives of the delta-function. These distributions are also multipoles and they are prescribed to any order. They represent an extended region of non-self-interacting uncharged or charged dust, shrunken to a single point in space.
We show that the uncharged dust stress-energy multipole is divergence-free, while the divergence of the charged dust stress-energy multipole is given by the current and the external electromagnetic field. We show that they can be obtained by squeezing a regular dust stress-energy tensor onto the worldine. We discuss the aforementioned calculations in a coordinate-free manner.
\end{abstract}


\section{Introduction}
\label{ch_Intro}

There is much interest currently about distributional sources of gravity~\cite{Gratus:2020cok,Gratus:2022vhm,Gratus:2023zcm,Gratus:2018kyo,Steinhoff:2014kwa,Steinhoff:2012rw,Steinhoff:2009tk},  in particular with reference to sources of gravitational waves. These are sources of gravity where all the mass is concentrated on a worldline. Hence the stress-energy tensor is the integral of a delta-function, and its derivatives are along the worldline. One may consider such distributional sources as approximations where the spatial extent of the source is small compared to the observers distance from such a source.

Since Einstein's equation are non-linear equations it is not possible to directly equate the Einstein tensor and a distributional source. By contrast, it is possible to let this distribution be the source for the linearised Einstein's equations. The solutions to the linearised equations can naturally be interpreted as gravitational waves. Hence we can interpret the distribution stress-energy tensor as sources for gravitational waves. In order to be a source of gravity or gravitational waves the stress-energy tensor must satisfy two conditions, namely being symmetric and divergence-free. These conditions can be relaxed if one is only considering a partial stress-energy tensor. For example, if the total stress-energy tensor has two components (one for matter and the other for the electromagnetic field) then it is only the total stress-energy tensor which needs to be symmetric and divergence-free.

In~\cite{Gratus:2020cok} the authors briefly look at the (uncharged) dust model for a quadrupole stress-energy tensor. In this article we extend the work. 
We posit the distributional dust stress-energy tensor. This tensor has not, as far as the authors are aware, been considered before, other than the brief mention in \cite{Gratus:2020cok}. We look at this distributional dust stress-energy tensor in detail, showing it is symmetric and divergence-free for all orders. 

We then consider the distributional stress-energy tensor for charge dust which is symmetric, but the divergence is not zero. As such this can only be a partial stress-energy tensor representing the matter in the model. It should be added to the stress-energy tensor of the electromagnetic field. This is achieved in \cite{Gratus:2021izz} for many charged particles, at the monopole order, where each particle responds to the fields of the other particles. However, this is not possible here due to the rapid diverging of the electromagnetic fields as one approaches the worldlines. As a result, we only demand that the distribution interacts with an external electromagnetic field, and we derive the corresponding divergence equation that it must satisfy. Again we posit an original distributional stress-energy tensor which satisfies this divergence condition.

Distributions can also be be considered as multipoles. The order of the multipole is defined as the maximum number of derivatives of the delta-function used to define it. With respect to sources of gravitational waves, the most interesting case is that of the quadrupole. As a heuristic argument, one can say that the monopole and dipole do not give rise to any gravitational waves, whereas for orders above the quadrupole the corresponding gravitational waves fall off with distance at a faster rate. With current technology it is already challenging to detect the quadrupole contribution, so these higher moments are not relevant. Thus the dominant contribution to gravitational waves is the quadrupole moment. In the case when the background metric is Minkowski, there is an explicit formula for the components of the gravitational waves in terms of the moments of the quadrupole~\cite{Gratus:2023zcm}.

The monopole has no derivatives of the delta-functions. The symmetry and divergence-free conditions imply that the worldline must be a geodesic and that mass is conserved. In the charged case, it implies that the worldline satisfies the Lorentz force equation. 

The dipole has a single derivative of the delta-function. If the worldline is prescribed, the components of the uncharged dipole satisfy the Mathisson–Papapetrou–
Tulczyjew–Dixon equations. This is a well defined system and the dynamics of the components are completely determined by the initial values. By contrast, if the worldline is not prescribed then there is an under-determined system~\cite{Gratus:2020cok,Steinhoff:2009tk} and additional equations are required to determine the motion of the worldline and the dynamics of the dipole. The same problem occurs if the dipole is charged, especially if it is constructed from multiple species.

The quadrupole has two derivatives of the delta-function. In this case, even for uncharged source with the worldline prescribed this is an under-determined system. There are 40 ordinary differential equations (ODEs), for 60 components. Thus one observes that for the most important case, namely the quadrupole, it is not possible to calculate the dynamics of the moments without additional information. These additional pieces of information are called {\em constitutive relations} as they are determined by the underlying constituents of the source. This is to be expected as the gravitational waves arising from two orbiting neutron stars, would be distinct from that of an asymmetric supernova. 

The challenge addressed in this article is to derive the dynamics of multipoles representing either charged or uncharged dust. Here the uncharged dust can model a low density of matter which only interacts with an external gravitational field. It does not model a distribution of matter which is bound by its gravitational field such as orbiting neutron stars. 
One of the consequences of such a model, which we show here, is that it does not spin. This is in line with our intuition, as a distribution of non-interacting dust would fly apart instead of spinning. In~\cite{Gratus:2020cok} we conjectured the constitutive relations for a dust model. This included a non zero spin component, and so does not correspond to the dust multipoles presented here.

The distributional charged dust models dust which interacts with an external electromagnetic field, not its own internal field. Thus it cannot be used to model a body held together by its own electrostatic forces. The external electromagnetic field in interstellar space is only of the order a few microgauss. By contrast electromagnetic fields near planets are 10s of Gauss and those near a neutron star or black hole may be 1000s of Gauss. Thus the charged dust distribution can be used to model matter orbiting a neutron star or in the accretion disc of a black hole.

\vspace{1em}

This article is arranged as follows. In section \ref{ch_Stress} we recap the Ellis representation~\cite{Ellis:1975rp} of a multipole. We state the dynamic equations for the quadrupole total stress-energy tensor. In~\cite{Gratus:2023zcm} the authors compared the advantages of the Ellis representation which uses partial derivatives and the Dixon representation~\cite{Dixon:1970zz} which uses covariant derivatives. 

In section \ref{ch_Dust} we look at the uncharged dust multipole. Using the Ellis representation, in a coordinate system adapted to a congruence of geodesics, allows us to greatly simplify the calculations. We present the uncharged dust multipole for any order and show that is it divergence-free. We also show that it automatically satisfies the dynamic equations for a total stress-energy tensor. We also show that it arises when one squeezes a regular dust stress-energy tensor onto a worldline.

In section \ref{ch_CDust} we repeat the process for a charged dust multipole. In this case we use a coordinate system adapted to a congruence of worldlines which satisfy the Lorentz force equation (for the same species).  We derive the formula for divergence of the stress-energy tensor for a charged distribution for which there is no self interaction. We present the charged dust multipole for any order and show that its divergence satisfies this formula. We then derive the dynamical equations for the quadrupole moments of an arbitrary charged quadrupole,  and show that it is satisfied by the dust quadrupole.

In section \ref{ch_DeRham} we show how the above calculations can be performed in a coordinate free manner, using the exterior covariant derivative; this is useful when expressing distributional quantities in coordinate systems not adapted to the flow. Arbitrary uncharged and charged multipoles up to quadrupole order were considered, and the equations for the components were derived. 

Finally, in chapter \ref{ch_Concl} we conclude and discuss future work.


\section{The stress-energy tensor in adapted Ellis coordinates}
\label{ch_Stress}

Let $(M,g)$ be spacetime with the Levi-Civia connection. We use Greek indices for the range $\mu,\nu,\ldots=0,1,2,3$ and Latin indices for $a,b,\ldots=1,2,3$, with implicitly summation for repeated indexes.  Round brackets in the indices mean the complete symmetric sum of these indices, for example 
$\chi^{\mu\nu(a b c)}=\tfrac16(\chi^{\mu\nu a b c}+\chi^{\mu\nu a c b}+\chi^{\mu\nu b a c}+\chi^{\mu\nu b c a}+\chi^{\mu\nu c a b)}+\chi^{\mu\nu c b a)})$.

Since we are dealing with distributions it is most convenient to consider $T^{\iMa\iMb}$ as a tensor density\footnote{An integral over $\Mman$ must contain the measure $\omega$. There is therefore the following choice: one can choose $T^{\iMa\iMb}$ or $\phi_{\iMa\iMb}$ to be a density of weight 1, or put $\omega$ explicitly in the integrand.  Here we have chosen to make $T^{\iMa\iMb}$ a density.} of weight 1. Thus $\omega^{-1} T^{\iMa\iMb}$ is a tensor, where
\begin{align}
\omega=\sqrt{-\det(g_{\iMa\iMb})}
\DEfullstop
\label{Intro_def_rootg}
\end{align}
 The definition of the covariant derivative of a tensor $S^{\iMa\iMb\cdots}$ density of weight 1 is given by
\begin{align}
\nabla_\iMa S^{\iMb\iMc\cdots}
=
\omega \nabla_\iMa (\omega^{-1} S^{\iMb\iMc\cdots})
=
- 
\Gamma^\iMd_{\iMa\iMd}\, S^{\iMb\iMc\cdots} 
+
\partial_\iMa S^{\iMb\iMc\cdots} 
+
\Gamma^\iMb_{\iMa\iMd}S^{\iMd\iMc\cdots} 
+
\Gamma^\iMc_{\iMa\iMd}S^{\iMb\iMd\cdots} 
+\cdots
\label{Intro_ten_den}
\end{align}
so that if $U^\iMa$ is a density of weight 1
\begin{align}
\nabla_\iMa U^\iMa 
=
\partial_\iMa U^\iMa 
\label{Intro_ten_res}
\DEfullstop
\end{align}
In this article all distributions are considered to be Schwartz distributions. The stress-energy tensor $T^{\mu\nu}$ density  distribution satisfies the symmetry condition
\begin{align}
T^{\mu\nu} = T^{\nu\mu}
\label{Stress_Sym}
\end{align}
and the divergence-free condition
\begin{align}
\nabla_\mu T^{\mu\nu} = 0
\label{Stress_Div-free}
\DEfullstop
\end{align}
It is defined by the way it acts on test tensors $\phi_{\mu\nu}$ of compact support via
\begin{align}
\int_M T^{\mu\nu}\,\phi_{\mu\nu} \, d^4 x
\label{Stress_Tmunu_action}
\DEfullstop
\end{align}
There are several ways of writing the distributional stress-energy tensor. These are given by the Ellis representation in general coordinates, the Ellis representation in adapted coordinates, and the Dixon representation. There is also a coordinate free construction. 
For this work the Ellis representation in adapted coordinates greatly simplifies the calculations. 

Let $C^{\mu}(\sigma)$ be the worldline which is the support of $T^{\mu\nu}$. We work in a coordinate system $(\sigma,z^1,z^2,z^3)$ which is adapted to the worldline $C^{\mu}(\sigma)$, so that $C^{\mu}(\sigma)=(\sigma,0,0,0)$ and $\Cdot^\mu=\delta^\mu_0$. Let $\Vz=(z^1,z^2,z^3)$ denote the spatial coordinates.
In these adapted Ellis coordinates, the general multipole of order $k$  can be written\footnote{In this article we have slightly changed the notation compared to \cite{Gratus:2020cok}. We have removed the trailing zeros in $\chi^{\mu \nu a_1 \dots a_r}$. This simplifies the notation when dealing with arbitrary order.} as
\begin{align}
T^{\mu \nu}
=
\sum_{r=0}^{k} \frac{1}{r!} 
\chi^{\mu \nu a_1 \dots a_r } (\sigma)
\partial_{a_1} \dots \partial_{a_r} \delta ^{(3)}(\Vz)
\label{Stress_Tmunu_Multi}
\DEcomma
\end{align}
so that (\ref{Stress_Tmunu_action}) becomes
\begin{align}
\int_M T^{\mu\nu}\,\phi_{\mu\nu} d^4 x
=
\sum_{r=0}^{k} \frac{(-1)^r}{r!} 
\int_\Real \chi^{\mu \nu a_1 \dots a_r } (\sigma)
\ \partial_{a_1} \dots \partial_{a_r} \phi_{\mu\nu}(\sigma,\Vzero)
\label{Stress_Tmunu_res}
\DEfullstop
\end{align}
From the symmetry (\ref{Stress_Sym}) these components satisfy
\begin{align}
\chi^{\mu \nu a_1 \dots a_r }
=
\chi^{\nu \mu a_1 \dots a_r }
\label{Stress_Sym_comp1}
\DEcomma
\end{align}
while from the commutation of partial derivatives we have
\begin{align}
\chi^{\mu \nu a_1 \dots a_r }
=
\chi^{\mu \nu (a_1 \dots a_r) }
\label{Stress_Sym_comp2}
\DEfullstop
\end{align}
By using squeezing, as we do below in section \ref{ch_Dust}, we see that there is a relationship between the components $\chi^{\mu \nu a_1 \dots a_r}$ and the moments of a regular stress-energy tensor.

At the quadrupole $k=2$ order (\ref{Stress_Tmunu_Multi}) becomes
\begin{align}
T^{\mu \nu}
=
\chi^{\mu \nu} (\sigma) \delta ^{(3)}(\Vz)
+
\chi^{\mu \nu a} (\sigma) \partial_{a} \delta ^{(3)}(\Vz)
+
\tfrac12\chi^{\mu \nu a b} (\sigma) \partial_{a} \partial_b \delta ^{(3)}(\Vz)
\label{Stress_Tmunu_quad}
\DEfullstop
\end{align}
From the divergence-free condition~\eqref{Stress_Div-free} these components satisfy 
\begin{align}
\dot\chi^{\mu0}
&=
- \Gamma^{\mu}_{\nu\rho}\, \chi^{\rho\nu}
+(\partz_{\iSa}\Gamma^{\mu}_{\nu\rho})\, \chi^{\rho \nu  \iSa}
-\tfrac12\big(\partz_{\iSb}\partz_{\iSa}\Gamma^{\mu}_{\nu\rho}\big) 
\chi^{\rho \nu \iSa \iSb}
\DEcomma
\label{QP_DTeqn_a000}
\\ 
\dot \chi^{\mu0\iSa}
&=
-\chi^{\mu\iSa}
- \Gamma^{\mu}_{\nu\rho}\, \chi^{\rho \nu \iSa}
+ (\partz_{\iSb}\Gamma^{\mu}_{\nu\rho})\, \chi^{\rho \nu \iSb \iSa}
\DEcomma
\label{QP_DTeqn_a00m}
\\ 
\dot \chi^{\mu 0 \iSa\iSb}
&=
- 2\chi^{\mu (\iSb\iSa) }
- \Gamma^{\mu}_{\nu\rho}\, \chi^{\rho \nu \iSa \iSb}
\label{QP_DTeqn_a0mn}
\DEnone
\end{align}
and
\begin{align}
\chi ^{\mu (\iSa \iSb \iSc)} &=0
\label{QP_DTeqn_alg}
\DEfullstop
\end{align}
This is proved in~\cite{Gratus:2020cok}. It is also a special case of theorem~\ref{thm Charged SE quad}, when $F_{\mu\nu}=0$ which is proved below.


\section{Uncharged dust}
\label{ch_Dust}

We express the formula for uncharged dust in an adapted coordinate system $(\sigma,z^1,z^2,z^3)$. This coordinate system is adapted to a congruence of worldline. Thus each curve by given $z^a=\text{const.}$ for $a=1,2,3$ is a geodesic. Hence the Christoffel symbols satisfy
\begin{align}
\Gamma^\mu_{00} = 0 
\label{Dust_Chrisoffel}
\DEfullstop
\end{align}
From \eqref{Intro_ten_res}, then setting $U^\mu=\delta^\mu_0$ to be a vector density of weight 1 we have
\begin{align}
    \nabla_\mu \delta^\mu_0 = \partial_\mu \delta^\mu_0 = 0
    \DEfullstop
    \label{Dust_Del_delta}
\end{align}

We can now formulate the dust multipole stress-energy tensor, in terms of this adapted coordinate system. As we stated in the introduction, this has not been considered previously in the literature, except for a brief mention in \cite{Gratus:2020cok}. This is a tensor densities of order $k$ and weight 1, given by
\begin{align}
T^{\mu\nu} 
&=
m\,\delta^\mu_0 \delta^\nu_0 \sum_{r=0}^k \frac{1}{r!} {{Y}}^{a_1\cdots a_r} 
\partial_{a_1} \cdots \partial_{a_r} \delta^{(3)}(\Vz)
\label{Dust_Tmunu_dust}
\DEcomma
\end{align}
where each ${{Y}}^{a_1\cdots a_r}$ is a constant and satisfies the symmetries (\ref{Stress_Sym_comp1}) and (\ref{Stress_Sym_comp2}), {and $Y^\emptyset=1$. Here $Y^\emptyset$ refers to case when there are no indices on $Y$, i.e. $r=0$. This mass could be incorporated into the $Y^{a_1\cdots a_r}$. However it is needed in the charged case when we need the ratio $q/m$.} We will show that this stress-energy tensor satisfies the divergence-free condition (\ref{Stress_Div-free}) and is also the limit of regular dust as it is squeezed onto the worldline.
\begin{lemma}
\label{lm_Dust_div-free}
The stress-energy distribution given in \eqref{Dust_Tmunu_dust} satisfies the divergence-free condition \eqref{Stress_Div-free}.
\end{lemma}
\begin{proof}
Since $T^{\mu\nu}$ is a tensor density of weight 1, we can choose any of the factors on the right hand side of \eqref{Dust_Tmunu_dust} to carry the tensor density. We choose the factor $\delta^\mu_0$ to have weight 1 and the rest of the factors to have weight 0. Thus $\nabla_\mu\delta^\mu_0=0$.
From (\ref{Dust_Chrisoffel}) 
\begin{align*}
\nabla_\mu T^{\mu\nu}
&=
m \delta^\mu_0 (\nabla_\mu \delta^\nu_0) \sum_{r=0}^k {\frac{1}{r!}} {{Y}}^{a_1\cdots a_r} 
\partial_{a_1} \cdots \partial_{a_r} \delta^{(3)}(\Vz)
\\&\quad +
 m\delta^\mu_0 \delta^\nu_0 \sum_{r=0}^k {\frac{1}{r!}} \partial_{\mu} \big({{Y}}^{a_1\cdots a_r} 
\partial_{a_1} \cdots \partial_{a_r} \delta^{(3)}(\Vz)\big)
\\&=
m \Gamma^\nu_{00} \sum_{r=0}^k {\frac{1}{r!}} {{Y}}^{a_1\cdots a_r} 
\partial_{a_1} \cdots \partial_{a_r} \delta^{(3)}(\Vz)
\\&\quad +
 m\delta^\nu_0 \sum_{r=0}^k {\frac{1}{r!}} \partial_{0} \big({{Y}}^{a_1\cdots a_r} 
\partial_{a_1} \cdots \partial_{a_r} \delta^{(3)}(\Vz)\big)
\\&\quad
=0
\DEfullstop
\end{align*}
\end{proof}

\begin{lemma}
\label{lm_Dust_check}
As a check, we can show that at quadrupole order the dust stress-energy tensor \eqref{Dust_Tmunu_dust} 
satisfies equation \eqref{QP_DTeqn_a000}--\eqref{QP_DTeqn_alg}.
\end{lemma}
\begin{proof}
From (\ref{Dust_Tmunu_dust}) we see 
$\chi^{\rho \nu \sigma \ldots}=\delta^\rho_0\delta^\nu_0 {{Y}}^{\sigma \ldots}$. 
Hence, from (\ref{Dust_Chrisoffel}), it is trivial to see that the right
hand sides of (\ref{QP_DTeqn_a000})-(\ref{QP_DTeqn_a0mn})
vanish. Likewise for the left hand side of (\ref{QP_DTeqn_alg}). 
Since $\chi^{\rho \nu \ldots}$ are constant the left hand side of
(\ref{QP_DTeqn_a000})-(\ref{QP_DTeqn_a0mn}) also vanish.
\end{proof}

We can understand this distribution as a model for dust by taking the squeezed limit.
Let 
\begin{equation}
\calT^{\mu\nu}
=
\varrho(\Vz) \delta^\mu_0 \delta^\nu_0
\label{Dust_reg_T}
\DEcomma
\end{equation}
where $\varrho$ is a scalar field (density of weight 0) and $\delta_0^\mu$ is a vector density of weight 1. That is $\partial_0(\varrho)=0$. This is the stress-energy  tensor density for dust, and we see that 
\begin{align*}
\nabla_\mu \calT^{\mu\nu}
&\quad
=
\nabla_\mu \big(\varrho(\Vz) \delta^\mu_0 \delta^\nu_0\big)
\\&\quad
=
\delta^\nu_0 \delta^\mu_0 \nabla_\mu \varrho(\Vz) 
+
\delta^\nu_0 \varrho(\Vz) \nabla_\mu \delta^\mu_0
+
\varrho(\Vz) \delta^\mu_0 \nabla_\mu \big(\delta^\nu_0\big)
\\&\quad
=
\delta^\nu_0 \delta^\mu_0 \partial_\mu \big(\varrho(\Vz) \big)
+
\varrho(\Vz) \nabla_0 \big(\delta^\nu_0\big)
\\&\quad
=
\delta^\nu_0 \partial_0 \big(\varrho(\Vz) \big)
+
\varrho(\Vz) \Gamma^\mu_{00}
\\&\quad
= 0
\DEfullstop
\end{align*}
Now consider a 1--parameter family of such stress-energy tensor densities, of weight 1, given by 
\begin{align}
\calT^{\mu\nu}_\epsilon
=
\epsilon^{-3}
\varrho(\epsilon^{-1} z) \delta^\mu_0 \delta^\nu_0 
\label{Dust_reg_T_eps}
\DEfullstop
\end{align}

\begin{lemma}
\label{lm_Dust_expan}
The Taylor expansion about $\epsilon=0$, to order $k$, is given by
\begin{equation}
\calT^{\mu\nu}_\epsilon
=
\hat{T}^{\mu\nu}_\epsilon
+ O(\epsilon^{k+1})
\label{Dust_T_expan}
\DEcomma
\end{equation}
where
\begin{equation}
\hat{T}^{\mu\nu}_\epsilon
=
m\,\delta^\mu_0 \delta^\nu_0 \sum_{r=0}^k 
\frac{\epsilon^r}{r!}
{{Y}}^{a_1\cdots a_r} 
\partial_{a_1} \cdots \partial_{a_r} \delta^{(3)}(\Vz)
\label{Dust_def_hat_TT}
\DEcomma
\end{equation}
and
\begin{align}
{{Y}}^{a_1\cdots a_r}
=
\frac{(-1)^r}{m}\int_{\Real^3}
z^{a_1}\cdots z^{a_r} \varrho(\Vz)\, d^3\Vz
\label{Dust_moments}
\end{align}
and the symbol $O(\epsilon^{k+1})$ means that any difference falls to zero as fast as $\epsilon^{k+1}$.
\end{lemma}
\begin{proof}
This follows from setting
$w^\iSa=z^\iSa/\varepsilon$ and Taylor expanding around $\varepsilon=0$
we have
\begin{align*}
\int_{\Real^4}\TReg^{\mu\nu}_\varepsilon(&\sigma,\Vz) \,
\phi_{\mu\nu}(\sigma,\Vz)\, 
d\sigma\,d^3z
\\&=
\int_{\Real} d\sigma \int_{\Real^3} d^3z\ 
\TReg^{\mu\nu}_\varepsilon(\sigma,\Vz) \,\phi_{\mu\nu}(\sigma,\Vz)\, 
\\&=
\int_{\Real} m \ d\sigma \int_{\Real^3} d^3z\ \frac{1}{m}
\epsilon^{-3}
\varrho(\epsilon^{-1} z) \delta^\mu_0 \delta^\nu_0 
\,\phi_{\mu\nu}(\sigma,\Vz)\, 
\\&=
\int_{\Real} m \ d\sigma \int_{\Real^3} d^3\Vw\ \frac{1}{m}
\varrho(\Vw) 
\,\phi_{00}(\sigma,\varepsilon \Vw)
\\&=
\int_{\Real} m \ d\sigma \int_{\Real^3} d^3\Vw\ \frac{1}{m} 
\sum_{r=0}^k \frac{\epsilon^r}{r!} \varrho(\Vw)
\,w^{a_1}\cdots w^{a_r}
\big(\partz_{a_1}\cdots\partz_{a_r}\phi_{00}(\sigma,\Vzero)\big)
+ O(\epsilon^{k+1})
\\&=
\sum_{r=0}^k\frac{\epsilon^r}{r!}
\int_{\Real} m \ d\sigma \big(\partz_{a_1}\cdots\partz_{a_r}\phi_{00}(\sigma,\Vzero)\big)
\int_{\Real^3} d^3\Vw\ \frac{1}{m}
w^{a_1}\cdots w^{a_r}\,\varrho(\Vw)
+ O(\epsilon^{k+1})
\\&=
\sum_{r=0}^k\frac{\epsilon^r (-1)^r}{r!}
\int_{\Real} d\sigma \ m {{Y}}^{a_1\cdots a_r} \big(\partz_{a_1}\cdots\partz_{a_r}\phi_{00}(\sigma,\Vzero)\big)
+ O(\epsilon^{k+1})
\\&=
\sum_{r=0}^k \frac{\epsilon^r}{r!} 
\int_{\Real} m \ d\sigma  \int_{\Real^3} d^3\Vz {{Y}}^{a_1\cdots a_r}
\big(\partz_{a_1}\cdots\partz_{a_r}\delta^{(3)}(\Vz)\big)  \big(\phi_{00}(\sigma,\Vz)\big)
+ O(\epsilon^{k+1})
\\&=
\int_{\Real} d\sigma  \int_{\Real^3} d^3\Vz
\sum_{r=0}^k \frac{\epsilon^r}{r!} 
 m {\delta^\mu_0 \delta^\nu_0 {Y}}^{a_1\cdots a_r}
\big(\partz_{a_1}\cdots\partz_{a_r}\delta^{(3)}(\Vz)\big)  \big(\phi_{\mu\nu}(\sigma,\Vz)\big)
+ O(\epsilon^{k+1})
\\&=
\int_{\Real^4} d\sigma\,  d^3\Vz\ 
{\hat T}^{\mu\nu}_\epsilon\,\phi_{\mu\nu}(\sigma,\Vz)
+ O(\epsilon^{k+1})
\DEfullstop
\end{align*}
\end{proof}
Clearly setting $\epsilon=1$ we have $\hat{T}^{\mu\nu}_1=T^{\mu\nu}$. However the nature of \eqref{Dust_T_expan} is more subtle, since we cannot simply set $\epsilon=1$. There are various interpretations. One option is to choose a total error ${\cal E}_{\max}$. Then from \eqref{Dust_T_expan} there is a value of $\epsilon$ such that 
$ | {\cal T}^{\mu\nu}_{\epsilon} - {\hat T}^{\mu\nu}_{\epsilon} | <  {\cal E}_{\text{max}} $, for all components. One can then redefine the $Y^{a_1\cdots a_r}\to \epsilon^r Y^{a_1\cdots a_r}$ to incorporate this value of $\epsilon$. Then 
$ | {\cal T}^{\mu\nu}_{\epsilon} - T^{\mu\nu} | <  {\cal E}_{\text{max}} $. Furthermore by replacing $\epsilon \to\epsilon/2$ we reduce the error by ${\cal E}_{\max}\to 2^{-k-1}{\cal E}_{\max}$.



\begin{figure}
\centering\small
\begin{tikzpicture}[xscale=0.7,yscale=0.7]
\fill [green!60!white] (-2.5,0) rectangle (2.5,10) ;
\fill [green!80!black] (-1.5,0) rectangle (1.5,10) ;
\fill [green!50!black] (-1,0) rectangle (0.25,10) ;
\fill [green!50!black] (0.5,0) rectangle (1.25,10) ;

\draw [ultra thick,black](0,0) node[below] {$z{=}0$} -- +(0,10) ;
\draw [very thick,black](1,0) node[below]{1} -- +(0,10) ;
\draw [very thick,black](2,0) node[below]{2} -- +(0,10) ;
\draw [very thick,black](3,0) node[below]{3} -- +(0,10) ;
\draw [very thick,black](-1,0) node[below]{-1} -- +(0,10) ;
\draw [very thick,black](-2,0) node[below]{-2} -- +(0,10) ;
\draw [very thick,black](-3,0) node[below]{$z{=}$-3} -- +(0,10) ;
\draw [thick,->] (-3.8,3) -- node[left] {$\sigma$} +(0,4) ;

\draw [ultra thick,red,dashed] 
(-3.5,1) -- +(7,0) 
(-3.5,3) -- +(7,0) 
(-3.5,5) -- +(7,0)  
(-3.5,7) -- +(7,0)  ; 
\end{tikzpicture}
\quad
\begin{tikzpicture}[xscale=0.6]
\fill [green!60!white] (-2.5,0) to[out=90,in=-40] 
(-5,7) -- (5,7) to[out=-130,in=90] (2.5,0) -- cycle ;
\fill [green!80!black] (-1.5,0) to[out=90,in=-60] 
(-3,7) -- (3,7) to[out=-125,in=90] (1.5,0) -- cycle ;
\fill [green!50!black] (-1,0) to[out=90,in=-70] 
(-2,7) -- (0.5,7) to[out=-95,in=90] (0.25,0) -- cycle ;
\fill [green!50!black] (0.5,0) to[out=90,in=-100] 
(1,7) -- (2.4,7) to[out=-115,in=90] (1.25,0) -- cycle ;

\draw [ultra thick,black](0,0) node[below] {$z{=}0$} -- +(0,7) ;
\draw [very thick,black](1,0) node[below]{1} -- +(0,7) ;
\draw [very thick,black](2,0) node[below]{2} -- +(0,7) ;
\draw [very thick,black](3,0) node[below]{3} -- +(0,7) ;
\draw [very thick,black](-1,0) node[below]{-1} -- +(0,7) ;
\draw [very thick,black](-2,0) node[below]{-2} -- +(0,7) ;
\draw [very thick,black](-3,0) node[below]{$z{=}$-3} -- +(0,7) ;
\draw [thick,->] (-5,1) -- node[left] {$\sigma'$} +(0,4) ;

\draw [ultra thick,red,dashed] 
(-4,1) -- (4,1) 
(-4.5,3) -- (4.5,3) 
(-5,5) -- (5,5)  ; 
\end{tikzpicture}
\caption{On the left hand side the coordinate system is adapted to the geodesic flow. In this case we derive (\ref{Dust_T_expan}). However, on the right hand side the coordinate system is not adapted to the geodesic flow, but is still adapted to the worldline at $z=0$. In this case we do not get (\ref{Dust_T_expan}).}
\label{fig_Dust_flow}
\end{figure}
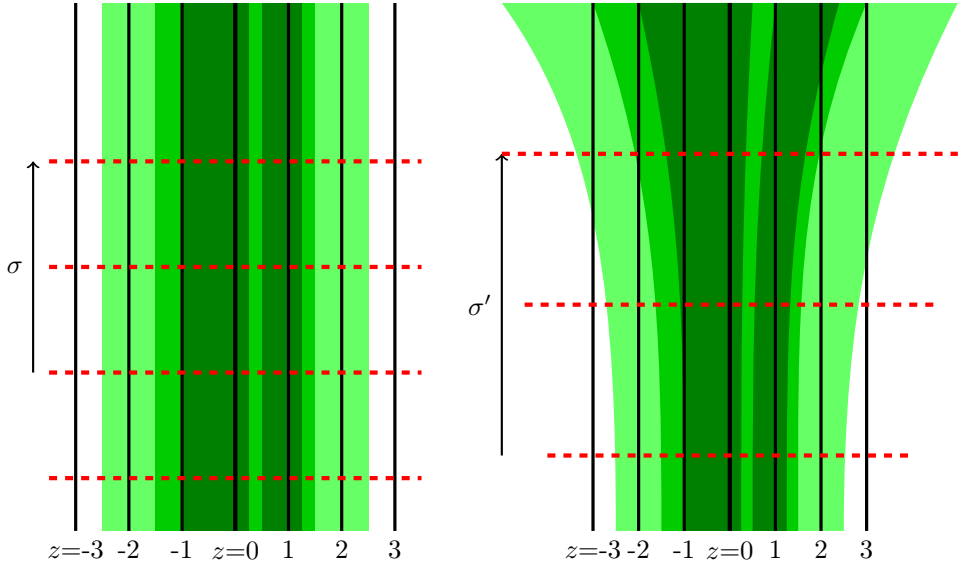

We observe that in the results of lemmas \ref{lm_Dust_div-free} and \ref{lm_Dust_check} the use of the coordinate system adapted to geodesic flow is purely for convenience and the result is independent of the coordinate system. This can be seen via the coordinate independent approach, given in section \ref{ch_DeRham} below. By contrast, the definition of $\calT^{\mu\nu}_\epsilon$ \eqref{Dust_reg_T_eps} depends on the coordinate system as seen in figure \ref{fig_Dust_flow}.


\section{Charged Dust}
\label{ch_CDust}

In contrast to uncharged dust, the stress-energy tensor is not divergence-free. This is because it is not a total stress-energy tensor. Instead we have
\begin{equation}
T_{\text{total}}^{\mu\nu}
=
T_{\text{mat}}^{\mu\nu} + T_{\text{EM}}^{\mu\nu}
\label{CDust__total}
\DEcomma
\end{equation}
where $T_{\text{total}}^{\mu\nu}$ is the total stress-energy tensor, and $T_{\text{mat}}^{\mu\nu}$ and $T_{\text{EM}}^{\mu\nu}$ are the contributions from the charged dust and the electromagnetic field. Since $\nabla_\mu T_{\text{total}}^{\mu\nu}=0$, then $\nabla_\mu T_{\text{mat}}^{\mu\nu}=-\nabla_\mu T_{\text{EM}}^{\mu\nu}$. 
Since for regular charged dust and electromagnetic field which are both smooth
\begin{align}
\nabla_\mu T_{\text{EM}}^{\mu\nu} 
=
- g^{\nu\rho} F^{\text{reg}}_{\rho\mu}\, J_{\text{reg}}^\mu\, 
\label{CDust_DivT_EM} 
\DEfullstop
\end{align}
Thus, we have
\begin{align}
\nabla_\mu T_{\text{mat}}^{\mu\nu} 
=
g^{\nu\rho} F^{\text{reg}}_{\rho\mu}\, J_{\text{reg}}^\mu\, 
\label{CDust_DivT_EM_Mat} 
\DEcomma
\end{align}
where the current $J_{\text{reg}}^\mu$ is given by Maxwell's equation
\begin{align}
J_{\text{reg}}^\mu = \nabla_\nu F_{\text{reg}}^{\nu\mu}
\label{CDust_Maxwell}
\DEfullstop
\end{align}
However, since we are dealing with Schwartz distributions, both $T^{\mu\nu}$ and $J^\mu$ are delta-functions which are infinite along the worldline. The problem is that, from (\ref{CDust_Maxwell}) the components of $F_{\rho\mu}$ also diverge as one approaches the worldline and thus (\ref{CDust_DivT_EM}) is not defined at the worldline. This leads to all the questions about what is the correct equation of motion when a point charged particle responds to its own electromagnetic field. In this article, we avoid this problem by making the $F_{\mu\nu}=F^{\text{Ext}}_{\mu\nu}$ an external electromagnetic field which does not satisfy (\ref{CDust_Maxwell}). Thus, we demand that a distributional stress-energy tensor $T^{\mu\nu}$, with a corresponding distributional current $J^\mu$, in the presence of an external electromagnetic $F^{\mu\nu}$ satisfies the divergence equation
\begin{align}
\nabla_\mu T^{\mu\nu} 
=
g^{\nu\rho} F_{\rho\mu}\, J^\mu\, 
\label{CDust_DivT} 
\DEfullstop
\end{align}

In~\cite{Gratus:2021izz} the problem of self interaction was solved by making each particle respond to the electromagnetic field of all the other particles. However, this approach relied on the fact that the components $F_{\mu\nu}\sim R^{-2}$ as one approached the worldline, where $R=\norm{\Vz}$ is the distance to the worldline.  However, since we are dealing with higher order multipoles then we would have $F_{\mu\nu}\sim R^{-k-2}$. This diverges to quickly and this approach will no longer work.

In Ellis representation and adapted coordinates the current is given by~\cite{Gratus:2020cok}
\begin{align}
J^\mu = \sum_{r=0}^k \frac{1}{r!} \gamma^{\mu a_1\cdots a_r}(\sigma)
\partial_{a_1} \dots \partial_{a_r} \delta ^{(3)}(\Vz)
\label{CDust_J_Multi}
\DEcomma
\end{align}
where $\gamma^{\mu a_1\cdots a_r}=\gamma^{\mu (a_1\cdots a_r)}$.
These are subject to the constraint arising from the conservation of charge
\begin{align}
\nabla_\mu J^\mu = 0
\label{CDust_Cons_charge}
\DEfullstop
\end{align}
At the octupole level 
\begin{align}
J^\mu = 
\gamma^{\mu}(\sigma)\delta ^{(3)}(\Vz)
+
\gamma^{\mu a}(\sigma)
\partial_{a} \delta ^{(3)}(\Vz)
+ \frac{1}{2}
\gamma^{\mu a b}(\sigma)
\partial_{a}\partial_b\delta ^{(3)}(\Vz)
+ \frac{1}{6}
\gamma^{\mu a b c}(\sigma)
\partial_{a}\partial_b\partial_c \delta ^{(3)}(\Vz)
\label{CDust_J_Octo}
\end{align}
and (\ref{CDust_J_Multi}) gives rise to the conditions
\begin{align}
\dot \gamma^{0} =0 \,,\quad
\dot \gamma^{0a} =-\gamma^{a} \,, \quad
\dot \gamma^{0 a b} = - 2\gamma^{(a b)} \,, \quad
\dot \gamma^{0 a b c} = - 3\gamma^{(a b c)} \,, \quad
\gamma^{(a b c d)} = 0
\label{CDust_Cons_charge_res}
\DEfullstop
\end{align}
The first of this equation implies the conservation of total charge
$\gamma^{0}  =q$. This is a very underdetermined system. At this octupole order there are $4\times(1+3+6+10)=80$ components with 15 algebraic equations, giving 65 unknowns. However, there are only 20 ODEs.

We wish to establish the general dynamic equations for the $\chi^{\mu\nu\dots}$ components of an arbitrary quadrupole stress-energy tensor. That is to generalise (\ref{QP_DTeqn_a000})-(\ref{QP_DTeqn_alg}). Since, in (\ref{CDust_DivT_EM}) we differentiate $T^{\mu\nu}$, to be most general we consider $J^\mu$ to be an octupole, $k=3$.
\begin{theorem}
\label{thm Charged SE quad}
The stress-energy quadrupole given by \eqref{Stress_Tmunu_quad} satisfies the divergence condition \eqref{CDust_DivT}, with current given by \eqref{CDust_J_Octo}, if and only if
\begin{align}
\label{oj11848n11}
\dot\chi^{\iMa}
+\Gamma^{\iMa}_{\iMb\iMc}\, \chi^{\iMc\iMb}
- \chi^{\iMc \iMb  \iSa}\,\partz_{\iSa}\Gamma^{\iMa}_{\iMb\iMc}
+\tfrac12
\chi^{\iMc \iMb \iSa \iSb}\partz_{\iSb}\partz_{\iSa}\Gamma^{\iMa}_{\iMb\iMc} 
& = \ga^{\rho } F^{\mu}{}_{\rho}   -\ga^{\rho a } \partial_{a}F^{\mu}{}_{\rho} +\tfrac{1}{2} \ga^{\rho a b }  \partial_{b}\partial_{a}F^{\mu}{}_{\rho} 
\nonumber \\
& \quad
-\tfrac{1}{6}\ga^{\rho a b c } \partial_{c} \partial_{b} \partial_{a} F^{\mu}{}_{\rho} \DEcomma
\\ \label{oj228848n22}
\dot \chi^{ \iMa 0 \iSa}
+\chi^{\iMa\iSa}
+ \Gamma^{\iMa}_{\iMb\iMc}\, \chi^{\iMc \iMb  \iSa}
-(\partz_{\iSb}\Gamma^{\iMa}_{\iMb\iMc})\, \chi^{\iMc \iMb \iSb \iSa} 
& = \ga^{\rho a }  F^{\mu}{}_{\rho}  - \ga^{\rho a b } \partial_{b}F^{\mu}{}_{\rho} \non
& \quad
+ \frac{1}{2}\ga^{\rho a b c }  \partial_{c}\partial_{b} F^{\mu}{}_{\rho}  \DEcomma  \\
\dot \chi^{\iMa 0 \iSa\iSb}
+ 2\chi^{\iMa (\iSb\iSa) } +\Gamma^{\iMa}_{\iMb\iMc}\, \chi^{\iMc \iMb \iSa \iSb} & = \ga^{\rho (a b) } F^{\mu}{}_{\rho}  -\ga^{\rho (a b) c }  \partial_{c}F^{\mu}{}_{\rho} \DEcomma
\label{QP_DTeqn_a0mn1} \\ \label{QP_DTeqn_a0mn19}
\chi^{\iMa (\iSa\iSb \iSc )} &=  \frac{1}{3}\ga^{\rho (a b c) } F^{\mu}{}_{\rho} 
\DEfullstop
\end{align}
\end{theorem}
\begin{proof}

We have that 
\begin{equation} \label{eq34}
\int_\Mman (\nabla_\iMa T^{\iMa\iMb} )\, \ToneTen_{\iMb}\,d^4 x  = \int_\Mman (g^{\nu\rho} F_{\rho\mu}\, J^\mu )\, \ToneTen_{\iMb}\,d^4 x 
\DEcomma
\end{equation}
where $\theta^{\nu}$ is a test vector.

Then
\begin{align*}
&   \int_\Mman (\nabla_\iMa T^{\iMa\iMb} )\, \ToneTen_{\iMb}\,d^4 x \non
& =
\int_\Mman 
\big(\partial_\iMa T^{\iMa\iMb} + 
\Gamma^\iMb_{\iMa\iMc} T^{\iMa\iMc} \big)\, 
\ToneTen_{\iMb}\,d^4 x \non
& =
\int_\Mman 
T^{\iMa\iMb} \big(\Gamma^\iMc_{\iMa\iMb} \,\ToneTen_{\iMc}-\partial_\iMa \ToneTen_{\iMb} \big)\, 
d^4 x
\\&=
\int_\Mman 
\Big(\chi^{\iMa\iMb} \,\deltaThree(\Vz)
+
\chi^{\iMa\iMb \iSa}\, \partz_\iSa \deltaThree(\Vz)
+
\tfrac12
\chi^{\iMa\iMb\iSa\iSb}\,
\partz_\iSa\partz_\iSb 
\deltaThree(\Vz)
\Big)
\big(\Gamma^\iMc_{\iMa\iMb} \,\ToneTen_{\iMc}-\partial_\iMa \ToneTen_{\iMb} \big)\, 
d^4 x
\\&=
\int_\Interval d\sigma\Big(
\chi^{\iMa\iMb}\,
\big(\Gamma^\iMc_{\iMa\iMb} \,\ToneTen_{\iMc}-\partial_\iMa \ToneTen_{\iMb} \big)
-
\chi^{\iMa\iMb\iSa}\,\partz_\iSa
\big(\Gamma^\iMc_{\iMa\iMb} \,\ToneTen_{\iMc}-\partial_\iMa \ToneTen_{\iMb} \big)
+
\tfrac12 
\chi^{\iMa\iMb\iSa\iSb}
\partz_\iSa \partz_\iSb
\big(\Gamma^\iMc_{\iMa\iMb} \,\ToneTen_{\iMc}-\partial_\iMa \ToneTen_{\iMb} \big)
\Big)
\\&=
\int_\Interval d\sigma\Big(
\chi^{\iMa\iMb}\,\Gamma^\iMc_{\iMa\iMb} \,\ToneTen_{\iMc}
-
\chi^{\iSa\iMb}\,\partial_\iSa \ToneTen_{\iMb} 
+
\dot\chi^{0\iMb}\,\ToneTen_{\iMb} 
\\&\qquad\qquad
-
\chi^{\iMa\iMb\iSa}\,\partz_\iSa
\big(\Gamma^\iMc_{\iMa\iMb} \,\ToneTen_{\iMc}\big)
+
\chi^{\iSb\iMb\iSa}\,\partz_\iSa
\partial_\iSb \ToneTen_{\iMb} 
-
\dot\chi^{0\iMb\iSa}\,\partz_\iSa
\ToneTen_{\iMb} 
\\&\qquad\qquad
+
\tfrac12 
\chi^{\iMa\iMb\iSa\iSb}
\partz_\iSa \partz_\iSb
\big(\Gamma^\iMc_{\iMa\iMb} \,\ToneTen_{\iMc}\big)
-
\tfrac12 
\chi^{\iSc\iMb\iSa\iSb}
\partz_\iSa \partz_\iSb
\partial_\iSc \ToneTen_{\iMb} 
+
\tfrac12 
\dot\chi^{0\iMb\iSa\iSb}
\partz_\iSa \partz_\iSb
\ToneTen_{\iMb} 
\Big)
\\&=
\int_\Interval d\sigma\Big(
\chi^{\iMa\iMb}\,\Gamma^\iMc_{\iMa\iMb} \,\ToneTen_{\iMc}
-
\chi^{\iSa\iMb}\,\partial_\iSa \ToneTen_{\iMb} 
+
\dot\chi^{0\iMc}\,\ToneTen_{\iMc} 
\\&\qquad\qquad
-
\chi^{\iMa\iMb\iSa}\,(\partial_\iSa\Gamma^\iMc_{\iMa\iMb}) \,\ToneTen_{\iMc}
-
\chi^{\iMa\iMb\iSa}\,
\Gamma^\iMc_{\iMa\iMb} \,\partial_\iSa\ToneTen_{\iMc}
+
\chi^{\iSb\iMb\iSa}\,\partial_\iSa
\partial_\iSb \ToneTen_{\iMb} 
-
\dot\chi^{0\iMb\iSa}\,\partial_\iSa
\ToneTen_{\iMb} 
\\&\qquad\qquad
+
\tfrac12 
\chi^{\iMa\iMb\iSa\iSb}
\big(\partial_\iSa \partial_\iSb
\Gamma^\iMc_{\iMa\iMb}\big)\ToneTen_{\iMc}
+
\chi^{\iMa\iMb\iSa\iSb}
\big(\partial_\iSa 
\Gamma^\iMc_{\iMa\iMb} \big)
\,\big(\partial_\iSb\ToneTen_{\iMc}\big)
+
\tfrac12 
\chi^{\iMa\iMb\iSa\iSb}
\Gamma^\iMc_{\iMa\iMb} \partial_\iSa \partial_\iSb
\ToneTen_{\iMc}
\\&\qquad\qquad
-
\tfrac12 
\chi^{\iSc\iMb\iSa\iSb}
\partial_\iSa \partial_\iSb
\partial_\iSc \ToneTen_{\iMb} 
+
\tfrac12 
\dot\chi^{0\iMb\iSa\iSb}
\partial_\iSa \partial_\iSb
\ToneTen_{\iMb} 
\Big)
\\&=
\int_\Interval d\sigma\bigg(
\ToneTen_{\iMc}
\Big(\chi^{\iMa\iMb}\,\Gamma^\iMc_{\iMa\iMb} 
+
\dot\chi^{\iMc}
-
\chi^{\iMa\iMb\iSa}\,(\partial_\iSa\Gamma^\iMc_{\iMa\iMb})
+
\tfrac12 
\chi^{\iMa\iMb\iSa\iSb}
\big(\partial_\iSa \partial_\iSb
\Gamma^\iMc_{\iMa\iMb}\big)
\Big)
\nonumber \\
&\qquad\qquad
-\partial_\iSa \ToneTen_{\iMc} \Big(
\chi^{\iSa\iMc}\,
+
\chi^{\iMa\iMb\iSa}\,\Gamma^\iMc_{\iMa\iMb} 
+
\dot\chi^{0\iMc\iSa}
-
\chi^{\iMa\iMb\iSb\iSa}
\big(\partial_\iSb 
\Gamma^\iMc_{\iMa\iMb} \big)
\Big)
\nonumber \\
&\qquad\qquad
+\partial_\iSa\partial_\iSb \ToneTen_{\iMc} 
\Big(\chi^{\iSb\iMc\iSa}
+
\tfrac12 
\chi^{\iMa\iMb\iSa\iSb}
\Gamma^\iMc_{\iMa\iMb}
+
\tfrac12 
\dot\chi^{0\iMc\iSa\iSb}
\Big)
-
\tfrac12 
\chi^{\iSc\iMb\iSa\iSb}
\partial_\iSa \partial_\iSb
\partial_\iSc \ToneTen_{\iMb} 
\bigg) 
\DEfullstop
\end{align*}

Moreover,
\begin{align*}
&   \int_\Mman g^{\nu\rho} F_{\rho\mu}\, J^\mu\, \ToneTen_{\iMb}\,d^4 x 
\\
& = \int_\Mman g^{\nu\rho} F_{\rho\mu}\, \Big(\gamma^{\mu}(\sigma)\delta ^{(3)}(\Vz)
+
\gamma^{\mu a}(\sigma)
\partial_{a} \delta ^{(3)}(\Vz)
+ \frac{1}{2}
\gamma^{\mu a b}(\sigma)
\partial_{a}\partial_b\delta ^{(3)}(\Vz)
\non &\qquad\qquad\qquad \quad
+ \frac{1}{6}
\gamma^{\mu a b c}(\sigma)
\partial_{a}\partial_b\partial_c \delta ^{(3)}(\Vz)\Big)\, \ToneTen_{\iMb}\,d^4 x 
\\
& = \int_\Mman \, \Big(\gamma^{\mu}(\sigma)\delta ^{(3)}(\Vz)
+
\gamma^{\mu a}(\sigma)
\partial_{a} \delta ^{(3)}(\Vz)
+ \frac{1}{2}
\gamma^{\mu a b}(\sigma)
\partial_{a}\partial_b\delta ^{(3)}(\Vz)
\\
& \qquad\qquad\qquad\quad + \frac{1}{6}
\gamma^{\mu a b c}(\sigma)
\partial_{a}\partial_b\partial_c \delta ^{(3)}(\Vz)\Big)\, \Big(g^{\nu\rho} F_{\rho\mu} \ToneTen_{\iMb}\Big)\,d^4 x 
\\
& = \int_\Interval d\sigma \Big(\gamma^{\mu}g^{\nu\rho} F_{\rho\mu} \ToneTen_{\iMb}-g^{\nu\rho}\ga^{\mu a}\p_{a}(F_{\rho \mu})\ToneTen_{\iMb}-g^{\nu\rho}\ga^{\mu a}F_{\rho \mu}\p_{a}\ToneTen_{\iMb}
\\
& \qquad\qquad\qquad\quad +\frac{1}{2}g^{\nu\rho}\ga^{\mu a b}\p_{a}\p_{b}(F_{\rho \mu})\ToneTen_{\iMb}+\frac{1}{2}g^{\nu\rho}\ga^{\mu a b}F_{\rho \mu}\p_{a}\p_{b}\ToneTen_{\iMb}\\
& \qquad\qquad\qquad\quad + \frac{1}{2}g^{\nu\rho}\ga^{\mu a b}\p_{a}(F_{\rho \mu})\p_{b}\ToneTen_{\iMb}+\frac{1}{2}g^{\nu\rho}\ga^{\mu a b}\p_{b}(F_{\rho \mu})\p_{a}\ToneTen_{\iMb}\\
& \qquad\qquad\qquad\quad -\frac{1}{6}g^{\nu\rho}\ga^{\mu a b c}\p_{a}\p_{b}\p_{c}(F_{\rho \mu})\ToneTen_{\iMb}-\frac{1}{6}g^{\nu\rho}\ga^{\mu a b c}F_{\rho \mu}\p_{a}\p_{b}\p_{c}\ToneTen_{\iMb}\\
& \qquad\qquad\qquad\quad - \frac{1}{6}g^{\nu\rho}\ga^{\mu a b c}\p_{a}\p_{b}(F_{\rho \mu})\p_{c}\ToneTen_{\iMb}-\frac{1}{6}g^{\nu\rho}\ga^{\mu a b c}\p_{c}(F_{\rho \mu})\p_{a}\p_{b}\ToneTen_{\iMb}\\
& \qquad\qquad\qquad\quad -\frac{1}{6}g^{\nu\rho}\ga^{\mu a b c}\p_{b}\p_{c}(F_{\rho \mu})\p_{a}\ToneTen_{\iMb}-\frac{1}{6}g^{\nu\rho}\ga^{\mu a b c}\p_{a}(F_{\rho \mu})\p_{b}\p_{c}\ToneTen_{\iMb}\\
& \qquad\qquad\qquad\quad - \frac{1}{6}g^{\nu\rho}\ga^{\mu a b c}\p_{a}\p_{c}(F_{\rho \mu})\p_{b}\ToneTen_{\iMb}-\frac{1}{6}g^{\nu\rho}\ga^{\mu a b c}\p_{b}(F_{\rho \mu})\p_{a}\p_{c}\ToneTen_{\iMb}\Big)\\
& = \int_\Interval d\sigma \Big(\theta_{\rho}\Big(\gamma^{\mu}g^{\rho \nu} F_{\nu\mu}-g^{\rho \nu}\ga^{\mu a}\p_{a}(F_{\nu \mu})\\
& \qquad\qquad\qquad\quad + \frac{1}{2}g^{\rho \nu}\ga^{\mu a b}\p_{a}\p_{b}(F_{\nu \mu}) -\frac{1}{6}g^{\rho \nu}\ga^{\mu a b c}\p_{a}\p_{b}\p_{c}(F_{\nu \mu})\Big)\\
& \qquad\qquad\qquad\quad - \p_{a} \theta_{\rho}\Big(g^{\rho \nu}\ga^{\mu a}F_{\nu \mu}- \frac{1}{2}g^{\rho \nu}\ga^{\mu a b}\p_{b}(F_{\nu \mu})\\
& \qquad\qquad\qquad\quad -\frac{1}{2}g^{\rho \nu}\ga^{\mu b a}\p_{b}(F_{\nu \mu})+\frac{1}{6}g^{\rho \nu}\ga^{\mu a b c}\p_{b}\p_{c}(F_{\nu \mu})\\
& \qquad\qquad\qquad\quad + \frac{1}{6}g^{\rho \nu}\ga^{\mu b a c}\p_{b}\p_{c}(F_{\nu \mu}) +\frac{1}{6}g^{\rho \nu}\ga^{\mu c b a}\p_{b}\p_{c}(F_{\nu \mu})\Big) \\
& \qquad\qquad\qquad\quad + \p_{a}\p_{b} \theta_{\rho}\Big(\frac{1}{2}g^{\rho \nu}\ga^{\mu a b}F_{\nu \mu}-\frac{1}{6}g^{\rho \nu}\ga^{\mu a b c}\p_{c}(F_{\nu \mu})\\
& \qquad\qquad\qquad\quad - \frac{1}{6}g^{\rho \nu}\ga^{\mu  b a c}\p_{c}(F_{\nu \mu})-\frac{1}{6}g^{\rho \nu}\ga^{\mu c b a}\p_{c}(F_{\nu \mu})\Big) \\
& \qquad\qquad\qquad\quad - \frac{1}{6}\p_{a}\p_{b}\p_{c} \theta_{\rho}\Big(g^{\rho \nu}\ga^{\mu a b c}F_{\nu \mu}\Big)\Big)\\
& = \int_\Interval d\sigma \Big(\theta_{\rho}\Big(\gamma^{\mu}F^{\rho}{}_{\mu}-\ga^{\mu a}\p_{a}F^{\rho}{}_{\mu}\\
& \qquad\qquad\qquad\quad + \frac{1}{2}\ga^{\mu a b}\p_{a}\p_{b}F^{\rho}{}_{\mu} -\frac{1}{6}\ga^{\mu a b c}\p_{a}\p_{b}\p_{c}F^{\rho}{}_{\mu}\Big)\\
& \qquad\qquad\qquad\quad - \p_{a} \theta_{\rho}\Big(\ga^{\mu a}F^{\rho}{}_{\mu}- \ga^{\mu a b}\p_{b}F^{\rho}{}_{\mu}\\
& \qquad\qquad\qquad\quad +\frac{1}{2}\ga^{\mu a b c}\p_{b}\p_{c}F^{\rho}{}_{\mu} \Big) \\
& \qquad\qquad\qquad\quad + \p_{a}\p_{b} \theta_{\rho}\Big(\frac{1}{2}\ga^{\mu a b}F^{\rho}{}_{\mu}-\frac{1}{2}\ga^{\mu a b c}\p_{c}F^{\rho}{}_{\mu} \Big)\\
& \qquad\qquad\qquad\quad - \frac{1}{6}\p_{a}\p_{b}\p_{c} \theta_{\rho}\Big(\ga^{\mu a b c}F^{\rho}{}_{\mu}\Big)\Big)
\DEfullstop
\end{align*}

Thus, equating the two sides in~\eqref{eq34}, we obtain~\eqref{oj11848n11}~-~\eqref{QP_DTeqn_a0mn19}.

\end{proof}

It is trivial to see that if $F_{\mu \nu}=0$ in~\eqref{oj11848n11}-\eqref{QP_DTeqn_a0mn19} then we recover \eqref{QP_DTeqn_a000}-\eqref{QP_DTeqn_alg}. Thus there are 40 ODEs and $60$ $\chi^{\mu\nu\ldots}$ components as in the uncharged case. If both $\chi^{\mu\nu\ldots}$ and $\gamma^{\mu\ldots}$ are unknown, then combining with the conservation of charge we have 60 ODEs for 125 unknowns.

\subsection{Charged dust stress-energy tensor and current}

Here we formulate the charged dust multipole stress-energy tensor and current. We again work in an adapted coordinate system $(\sigma,z^1,z^2,z^3)$. However, this time each curve given by $\underline{z}=\text{constant}$ is a solution to the Lorentz force equation with the same ratio $q/m$.
Hence, the Christoffel symbols satisfy
\begin{align}
\Gamma^\mu_{00} = \frac{q}{m} F^{\mu}{}_0
\label{CDust_Chrisoffel}
\DEfullstop
\end{align}

The stress-energy tensor density (of weight 1) has the same structure as (\ref{Dust_Tmunu_dust}) but in this new coordinate system. That is
\begin{align}
T^{\mu\nu} 
=
m \,\delta^\mu_0 \delta^\nu_0 \sum_{r=0}^k 
\frac{1}{r!}
{{Z}}^{a_1\cdots a_r} 
\partial_{a_1} \cdots \partial_{a_r} \delta^{(3)}(\Vz)
\label{CDust_Tmu}
\DEfullstop
\end{align}
where $Z^\emptyset=1$ and the $Z^{a_1\cdots a_r}$ are constants. In this model the distribution of charge is the same as the distribution of matter. Thus we are considering only a single species of charged particle. 
The current density (of weight 1) is given by
\begin{align}
J^\mu = q \, \delta^\mu_0 \sum_{r=0}^k \frac{1}{r!} {{Z}}^{a_1\cdots a_r}
\partial_{a_1} \dots \partial_{a_r} \delta ^{(3)}(\Vz)
\label{CDust_Jdust_Multi}
\DEfullstop
\end{align}
We can see that (\ref{CDust_Jdust_Multi}) trivially satisfies the conservation of charge \eqref{CDust_Cons_charge}, since using $\nabla_\mu \delta_0^\mu=0$
\begin{align}
\nabla_\mu J^\mu
&=
q \delta^\mu_0 \sum \partial_\mu(\tfrac{1}{r!} {{Z}}^{a_1\cdots a_r}
\partial_{a_1} \dots \partial_{a_r} \delta ^{(3)}(\Vz))
\nonumber \\
&= q \sum \partial_0(\tfrac{1}{r!} {{Z}}^{a_1\cdots a_r}
\partial_{a_1} \dots \partial_{a_r} \delta ^{(3)}(\Vz)) \nonumber \\
&=0
\DEfullstop
\end{align}

\begin{theorem}
\label{lm_CDust_conserv}
The stress-energy tensor \eqref{CDust_Tmu} satisfies the divergence condition \eqref{CDust_DivT_EM_Mat} where the current is given by \eqref{CDust_Jdust_Multi}.
\end{theorem}
\begin{proof}
We assume that $\delta^\mu_0$ is the factor with the weight 1, so that $\nabla_\mu \delta_0^\mu = 0$.
\begin{align*}
\nabla_\mu T^{\mu\nu}
&=
m(\nabla_\mu \delta^\mu_0) \delta^\nu_0 \sum_{r=0}^k \frac{1}{r!} {{Z}}^{a_1\cdots a_r} 
\partial_{a_1} \cdots \partial_{a_r} \delta^{(3)}(\Vz)
\\
& \quad  + m
\delta^\mu_0 (\nabla_\mu \delta^\nu_0) \sum_{r=0}^k \frac{1}{r!} {{Z}}^{a_1\cdots a_r} 
\partial_{a_1} \cdots \partial_{a_r} \delta^{(3)}(\Vz)
\\
& \quad  + m
 \delta^\mu_0 \delta^\nu_0 \sum_{r=0}^k \frac{1}{r!} \partial_{\mu} \big({{Z}}^{a_1\cdots a_r} 
\partial_{a_1} \cdots \partial_{a_r} \delta^{(3)}(\Vz)\big)
\\&=
m\Gamma^\nu_{00} \sum_{r=0}^k \frac{1}{r!} {{Z}}^{a_1\cdots a_r} 
\partial_{a_1} \cdots \partial_{a_r} \delta^{(3)}(\Vz)
\\
& = m(q/m) F^{\nu}{}_{0}\sum_{r=0}^k \frac{1}{r!} {{Z}}^{a_1\cdots a_r} 
\partial_{a_1} \cdots \partial_{a_r} \delta^{(3)}(\Vz)\\
& = q F^{\nu}{}_{\mu}\delta^{\mu}_{0}\sum_{r=0}^k \frac{1}{r!} {{Z}}^{a_1\cdots a_r} 
\partial_{a_1} \cdots \partial_{a_r} \delta^{(3)}(\Vz)\\
& = q g^{\nu \rho}F_{\rho \mu}\delta^{\mu}_{0}\sum_{r=0}^k \frac{1}{r!} {{Z}}^{a_1\cdots a_r} 
\partial_{a_1} \cdots \partial_{a_r} \delta^{(3)}(\Vz)\\
& = g^{\nu \rho}F_{\rho \mu} J^{\mu}
\DEfullstop
\end{align*}
\end{proof}

The dust stress-energy tensor \eqref{CDust_Tmu} and current to quadrupole \eqref{CDust_J_Multi} order are given by
\begin{align}
T^{\mu\nu} = m \delta^\mu_0 \delta^\nu_0
\Big(
\delta(\Vz) +
{{Z}}^a \partial_{a} \delta(\Vz) +
{{Z}}^{ab} \partial_{a}\partial_b \delta(\Vz)
\Big)
\label{CDust_T_quad}
\end{align}
and
\begin{align}
J^{\mu} = q \delta^\mu_0
\Big(
\delta(\Vz) +
{{Z}}^a \partial_{a} \delta(\Vz) +
{{Z}}^{ab} \partial_{a}\partial_b \delta(\Vz)
\Big)
\label{CDust_J_quad}
\DEfullstop
\end{align}

\begin{lemma}
The charged dust stress-energy quadrupole \eqref{CDust_T_quad} and current quadrupole \eqref{CDust_J_quad} satisfy the ODEs
\eqref{oj11848n11}-\eqref{QP_DTeqn_a0mn19}.
\end{lemma}
\begin{proof}
We wish to verify that the equations~\eqref{oj11848n11}~-~\eqref{QP_DTeqn_a0mn19} hold. The $\dot{\chi}^{\mu\nu\ldots}$ terms vanish while the $\chi^{\mu\nu\ldots}$ terms are non-zero when the first two indices are equal to zero. The condition on the Christoffel symbols for charged dust
is~\eqref{CDust_Chrisoffel}.

As one can see in equations~\eqref{oj11848n11}~and~\eqref{oj228848n22}, there are exact cancellations between the gravitational and electromagnetic terms in the case of a dust model.

One should also keep in mind the fact
that, for a given multipole order, the $\chi$ and $\gamma$ terms encompass the same $X$ constants
since
\begin{equation}
\chi^{\mu \nu a_1\cdots a_r} = m \delta^{\mu}_{0} \delta^{\nu}_{0} {{Z}}^{a_1\cdots a_r}
\end{equation}
and
\begin{equation}
\gamma^{\mu a_1\cdots a_r} = q \delta^{\mu}_{0} {{Z}}^{a_1\cdots a_r}
\DEfullstop
\end{equation}



Moreover, differentiating~\eqref{CDust_Chrisoffel} with respect to a spatial component gives
\begin{equation} \label{981}
\partial_{b}(\Gamma^{\mu}_{0 0}) = (q/m)\partial_{b}(F^{\mu}{}_{0})
\end{equation}
and
\begin{equation} \label{9811}
\partial_{b}\partial_{c}(\Gamma^{\mu}_{0 0}) = (q/m)\partial_{b}\partial_{c}(F^{\mu}{}_{0}) 
\DEfullstop
\end{equation}

In equation~\eqref{QP_DTeqn_a0mn1}, exact cancellations also arise.
Both sides of equation~\eqref{QP_DTeqn_a0mn19} vanish as well.
\end{proof}

To demonstrate that this is charged dust, one can repeat the calculation in lemma \ref{lm_Dust_expan}. In addition to confirm the current one can squeeze the regular current given by
\begin{align}
{\cal J}^{\mu\nu}_\epsilon
=
\epsilon^{-3}
\varrho(\epsilon^{-1} z) \delta^\nu_0 \omega
\label{CDust_Reg_Curr}
\DEfullstop
\end{align}

\section{The coordinate free de Rham formulation of the stress-energy tensor}
\label{ch_DeRham}

The results in this article can all be reproduced in a coordinate free notation using the language of differential geometry and de Rham currents. This is very useful when one needs to express distributional quantities such as the current and stress-energy tensor in a coordinate system which is not adapted to the flow (e.g. a coordinate system adapted to the observer, rather than the source). The transformation of the components for these quantities under change of coordinates is complicated \cite{Gratus:2020cok,Gratus:2018kyo}, involving higher order derivatives and integrals. By using a coordinate free notation, the components in a preferred coordinate system can then be extracted.
The detail of how to construct the stress-energy distribution is given in \cite[Section 6]{Gratus:2020cok}. 

Even though all the work can be repeated in this language, here we only reproduce the key result, theorem \ref{lm_CDust_conserv}. In this section we will use the coordinate free covariant derivative $\Bnabla$, also defined in \cite[Section 6]{Gratus:2020cok}.

Recall the stress-energy vector valued distribution $\tau$, so that it acts on test tensors of type (0,2) as arguments. That is $\tau[\beta\otimes\alpha]\in\Real$. This is symmetric so that
\begin{align}
\tau[\beta\otimes\alpha]=\tau[\alpha\otimes\beta]
\label{CoFree_SE_tau_symm}
\end{align}
and the divergenceless condition is given by
\begin{align}
D\tau=0
\label{CoFree_SE_Dtau=0}
\DEcomma
\end{align}
where ($\theta$ is a 1--form valued scalar)
\begin{align}
(D\tau)[\theta] = -\tau[D\theta]
\label{CoFree_SE_def_Dtau}
\DEcomma
\end{align}
and
\begin{align}
(D\theta)(U,V) = (\Bnabla_V\theta):U
\label{CoFree_SE_def_Dphi}
\DEcomma
\end{align}
where $\alpha:U$ is the internal product between the 1--form $\alpha$
and the vector $U$.

The right hand side of (\ref{CDust_DivT_EM_Mat}) is written
$\calF\wedge J$ where $J$ is a current distribution and $\calF$
encodes the Maxwell 2--form $F$. For a test 1--form $\alpha$ 
\begin{align}
(\calF \wedge J)[\alpha] = - J[i_{\dual{\alpha}} F]
\label{DeRham_def_FJ}
\DEcomma
\end{align}
where $\dual{\alpha}$ is the metric dual of $\alpha$. Equation
(\ref{CDust_DivT_EM_Mat}) becomes
\begin{align}
D\tau = \calF \wedge J
\label{DustMult_Dtau_k}
\DEfullstop
\end{align}

In order to construct symmetric stress-energy tensors we introduce the symmetry operator, $\Sym$, where 
\begin{align}
\Sym(\alpha\otimes\beta)
=
\tfrac12 \alpha\otimes\beta + \tfrac12 \beta\otimes\alpha
\label{DeRham_def_sym}
\DEfullstop
\end{align}
\begin{lemma}
Let $\theta$ be a 1--form, then
\begin{align}
\big(\Sym\,D(\theta)\big) (V,U)
=
\big(\Bnabla_V \theta - \tfrac12 i_V d\theta\big) :U
\label{DeRham_sym_D}
\DEcomma
\end{align}
where $i_V$ is the internal contraction. 
\end{lemma}
\begin{proof}
\begin{align*}
2\,\Sym\,D(\theta) (U,V)
&=
D(\theta) (U,V)
+
D(\theta) (V,U)
=
\Bnabla_U \theta:V + \Bnabla_V \theta:U
\\&=
U({\theta : V}) - \theta:\Bnabla_U V + \Bnabla_V \theta:U
\\&=
L_U\theta : V - \theta:(\Bnabla_U V + [U,V] ) + \Bnabla_V \theta : U
\\&=
L_U\theta : V - \theta:\Bnabla_V U + \Bnabla_V \theta : U
=
i_V L_U\theta + V(\theta: U) + 2\Bnabla_V \theta : U
\\&=
i_V i_U d\theta + i_V d i_U \theta - V(\theta: U) + 2\Bnabla_V \theta : U
=
(- i_V d\theta + 2\Bnabla_V \theta) : U
\DEfullstop
\end{align*}
\end{proof}
We can write (\ref{DeRham_sym_D}) without the arbitrary vector $U$
using the slot notation as
\begin{align}
\big(\Sym\,D(\theta)\big) (V,-)
=
\Bnabla_V \theta - \tfrac12 i_V d\theta
\label{DeRham_sym_D_alt}
\DEfullstop
\end{align}

If $\kappa$ is a distribution which acts on test  tensors of type
(0,2)  then we define its symmetry as 
\begin{align}
\Sym (\kappa)[\phi] = \kappa[\Sym(\phi)] 
\label{DeRham_def_sym_dist}
\DEfullstop
\end{align}
Also if $\kappa$ is a distribution which acts on test 1--forms then we
define
\begin{align}
(\kappa\otimes V) [\alpha\otimes\beta] 
=
\kappa[(\beta:V)\alpha] 
\label{kj}
\DEfullstop
\end{align}

\def\VC{Y}

We can now define the stress-energy and current distribution for dust
in a coordinate free manner. Let $Y$ be a vector field such that the
multipole trajectory $C$ is an integral curve. 
Let $W_1,\ldots,W_k\in\Gamma TM$ be a set of vector fields such that
\begin{align}
[W_j,\VC]=0
\label{DeRham_Lie_Bracket}
\DEfullstop
\end{align}
Let the current for the dust multipole be given by
\begin{align}
J = q L_{W_1} \cdots L_{W_K}  C_\PF(1)  
\label{DustMult_J}
\end{align}
and the corresponding stress-energy multipole
\begin{align}
\tau =  \Sym \big(m L_{W_1} \cdots L_{W_k} C_\PF(1)\otimes \VC\big)
\label{DustMult_SE}
\DEfullstop
\end{align}
Here $C_\PF(1)$ is the de Rham push forward, given in \cite[Section 6]{Gratus:2020cok}.
By definition $\tau$ is symmetric. 
From linearity we can construct any
current and stress-energy tensor by adding together an arbitary
number of $J$ and $\tau$. 
To compare (\ref{DustMult_J}) and (\ref{DustMult_SE}) with
(\ref{CDust_Jdust_Multi}) and (\ref{CDust_Tmu}), we observe that we
set $\VC=\partial_0$ and the $W_j$ as the coordinate vectors
$\partial_{a_i}$. We then act on a test 2--form.
We show here it also satisfies
the divergence property (\ref{DustMult_Dtau_k}).

\begin{theorem}
The stress-energy multipole $\tau$ and current multipole $J$ 
given by \eqref{DustMult_SE} and \eqref{DustMult_J} satisfy the
divergence condition \eqref{DustMult_Dtau_k}.
\end{theorem}
\begin{proof}
Let $\theta$ be a test 1--form
\begin{align*}
D\tau[\theta] 
&=
-\tau[D(\theta)]
=
-\Sym \big(m L_{W_1} \cdots L_{W_k} C_\PF(1)\otimes \VC\big)
[D(\theta)]
\\&
=
-m \big(L_{W_1} \cdots L_{W_k} C_\PF(1)\otimes \VC\big)
[\Sym D(\theta)]
\\&
=
-m L_{W_1} \cdots L_{W_k} C_\PF(1)
[\Sym D(\theta) (\VC,-)]
\\&
=
-m L_{W_1} \cdots L_{W_k} C_\PF(1)
[\nabla_\VC \theta - \tfrac12 i_\VC d\theta]
\\&
=
(-1)^{k+1}m C_\PF(1)
[L_{W_k} \cdots L_{W_1} (\nabla_\VC \theta - \tfrac12 i_\VC  d\theta)]
\\&
=
(-1)^{k+1}m 
\int C^\star \big(L_{W_k} \cdots L_{W_1} (\nabla_\VC \theta - \tfrac12 i_\VC  d\theta)\big)
\\&
=
(-1)^{k+1}m 
\int d\sigma\,
C^\star \big(i_\Cdot
L_{W_k} \cdots L_{W_1} (\nabla_\VC \theta - \tfrac12 i_\VC  d\theta)\big)
\\&
=
(-1)^{k+1}m 
\int d\sigma\,
C^\star \big(i_\VC
L_{W_k} \cdots L_{W_1} (\nabla_\VC \theta - \tfrac12 i_\VC  d\theta)\big)
\\&
=
(-1)^{k+1}m 
\int d\sigma\,
C^\star \big(
L_{W_k} \cdots L_{W_1} (i_\VC \nabla_\VC \theta - \tfrac12 i_\VC i_\VC  d\theta)\big)
\\&
=
(-1)^{k+1}m 
\int d\sigma\,
C^\star \big(L_{W_k} \cdots L_{W_1} (i_\VC \nabla_\VC \theta )\big)
\\&
=
(-1)^{k+1}m 
\int d\sigma\,
C^\star \big(L_{W_k} \cdots L_{W_1} (L_\VC (\theta:\VC) - \theta:\nabla_\VC \VC )\big)
\\&
=
(-1)^{k+1}m 
\int d\sigma\,
C^\star \big(L_\VC L_{W_k} \cdots L_{W_1} (\theta:\VC)\big)
\\&\qquad +
(-1)^{k}q
\int d\sigma\,
C^\star \big(L_{W_k} \cdots L_{W_1} (\theta:\dual{i_\VC F} )\big)
\\&
=
(-1)^{k+1}m 
\int 
d C^\star \big(L_{W_k} \cdots L_{W_1} (\theta:\VC)\big)
\\&\qquad + 
(-1)^{k+1}q
\int d\sigma\,
C^\star \big(L_{W_k} \cdots L_{W_1} (i_\VC i_{\dual\theta}F )\big)
\\&
=
(-1)^{k+1}q
\int d\sigma\,
C^\star \big(i_\VC L_{W_k} \cdots L_{W_1} (i_{\dual\theta}F )\big)
\\&
=
(-1)^{k+1}q
\int 
C^\star \big(L_{W_k} \cdots L_{W_1} (i_{\dual\theta}F )\big)
\\&
=
(-1)^{k+1}q
C_\PF(1) [L_{W_k} \cdots L_{W_1} (i_{\dual\theta}F )]
\\&
=
- q L_{W_1} \cdots L_{W_k} C_\PF(1) [i_{\dual\theta}F ]
=
- J [i_{\dual\theta}F ]
=
\big(\calF\wedge J\big)[ \theta]
\DEfullstop
\end{align*}
\end{proof}


\section{Conclusion and discussion}
\label{ch_Concl}

In this paper we consider the stress-energy multipole for both charged and uncharged dust. These are distributions which have support on a worldline. We demand that both are symmetric and that uncharged dust satisfies the divergence-free condition, whilst the divergence of charged dust is related to the current and the external electromagnetic field.

The required divergence of the charged multipole \eqref{CDust_DivT_EM} is subtle. Since the electromagnetic field of the generated by a multipole would diverge on the worldline (and this divergence is very fast), we cannot simply equate the divergence of the multipole stress-energy tensor with the divergence of the electromagnetic stress-energy tensor. Instead, inspired by the divergence of the electromagnetic stress-energy tensor, we posit the required equation \eqref{Dust_Tmunu_dust}.

We formulate both the charged and uncharged stress-energy multipoles to arbitrary order. We show how they satisfy the required conditions and also how they arise naturally in the limit as one squeezes regular dust onto the worldline. These are particularly simple in the Ellis representation of multipoles, with coordinates adapted to a flow of geodesics or the Lorentz force equation. In this case the components are constants. 

Although the multipoles are simple in the adapted coordinate system, and therefore their properties hold in all coordinate systems, the formula for transformation between coordinate systems is complicated. They are not tensorial as they involve both higher derivatives and integration~\cite{Gratus:2020cok,Gratus:2018kyo}. For this reason, in section~\ref{ch_DeRham}, we also show how the general results can be demonstrated in a coordinate free language.

We consider arbitrary uncharged and charged multipoles up to quadrupole order and derive the equations for the components. Then as a sanity check, we confirm that the components of the charged and uncharged dust multipoles, when truncated to quadrupoles, do indeed satisfy these equations.

In~\cite{Gratus:2020cok}, we observed that, at the quadrupole order, the divergence equations, are not sufficient to determine the dynamics of the components. For both cases, there are 40 equations for 60 variables, assuming the worldline is prescribed. Thus there is a need for constitutive relations to fully describe the dynamics. It is hoped that by deriving the equations of motion for well known matter, one could identify the constitutive relations. Unfortunately, although we have the dynamical equations for the components, and we demonstrate that they satisfy the required ODEs, it is not obvious how to identify particular equations as constitutive relations.

As stated in the introduction, this work can be applied to various branches of physics. The uncharged dust is a good model nebula, or even galactic systems. The dynamics of the quadrupole components are directly related to gravitational waves~\cite{Gratus:2023zcm}. Thus, one can compare the detected gravitational waves with those by the dust model. For example, when we are able to detect primordial gravitational waves we could ask if these are consistent with dust quadrupoles. One mathematical calculation which would need to be performed is to express the components in a coordinate system adapted to us as observers, instead of the geodesic flow of the source. In this context the coordinate free language will be invaluable. 

There are many other stress-energy multipoles one could consider. Examples include kinetic, pressure and spin.

In a kinetic model, there is a range of velocities at each event in spacetime, and one must work in 7-dimensional time-phase space. This range of velocities is incorporated into the kinetic ``distribution''\footnote{A different use of the word {\em distribution}.} scalar field on 7--dimensional phase-space-time space. 
Collisionless charged particles obey  
the Vlasov equation~\cite{vlasov}, which describes the time evolution of the kinetic distribution function of plasma, consisting of charged particles (electrons and ions) with long-range interaction. 
In~\cite{Dymnikov:1978,Channell:1983,Warwick:2023clf}, the dynamics of the components of a Vlasov multipole on phase-space-time space are given. 
Using the Ellis
representation of the de Rham current representation of
the moments, coordinate transformations were derived \cite{Warwick:2023clf}
between frames that mix the space and time coordinates. The results were
confirmed numerically for the case of particles orbiting a black
hole. The current and stress-energy distributions, corresponding to the Vlasov distribution, can be derived by projecting the distribution onto spacetime using the de Rham push forward.

Another possibility is to consider a fluid with a pressure. This model should arise in the limit as one squeezes a fluid with pressure onto the worldline. However, this cannot be done naively as, unlike dust, the pressure, directly opposes such squeezing. We conjecture that it would be possible if in~\eqref{Dust_T_expan}, the pressure acts at order $\epsilon^2$ and higher.

As noted in the introduction, cosmic dust does not possess any total spin. In order to introduce spin, one could look into the Weyssenhoff dust model~\cite{Weyssenhoff:1947,Obukhov:1987yu}, which includes a factor of spacetime torsion to model spin. Employing the adapted coordinate system, one may be able to compute the dynamics of the moments of
the Weyssenhoff dust multipole. 

\subsection*{Funding statement}
JG is grateful for
the support provided by STFC (the Cockcroft Institute ST/P002056/1 and
ST/V001612/1). 
ST is grateful for the support of Lancaster University's Faculty of Science and Technology. WS is also grateful for the support of Lancaster University's Faculty of Science and Technology.

\subsection*{Author Contributions}
All authors did the research regarding the uncharged dust and the coordinate free calculations.
JG and ST did the research regarding charged dust.
WS reviewed the Weyssenhoff dust model.
JG and ST wrote the manuscript.

\subsection*{Conflict of interest statement}
There are no conflicts of interest.

\subsection*{Ethics statement}
This work did not involve human participants, animal subjects, or sensitive data, and therefore did not require ethical approval. No ethical issues were identified in the course of this research.

\subsection*{Data Access Statement}
No additional data was created for this article.


\end{document}